\documentclass[10pt]{amsart}
\usepackage[utf8]{inputenc}
\usepackage{amscd}
\usepackage{amssymb}
\usepackage{pstricks}
\usepackage{graphicx}
\usepackage{float}
\usepackage{todonotes}
\usepackage{subcaption}
\usepackage{xcolor}
\usepackage{textcomp}
\usepackage{cancel}
\usepackage{comment}
\usepackage[mathlines]{lineno}
\usepackage{todonotes}
\newtheorem{theorem}{Theorem}
\newtheorem{proposition}{Proposition}

\newtheorem{definition}{Definition}

\newtheorem{remark}{Remark}

\numberwithin{equation}{section}

\title[Generalization of the Levins metapopulation model]{A generalization of the Levins metapopulation model to time varying colonization and extinction rates}
\title[Generalization of the Levins metapopulation model]{Generalizing the Levins metapopulation model to time varying colonization and extinction rates}

\author[Robledo]{Gonzalo Robledo}
\author[Bustamante]{Ramiro Bustamante}
\email{grobledo@uchile.cl,rbustama@uchile.cl}
\address{Departamento de Ciencias Ecol\'ogicas, Instituto de Ecolog\'ia y Biodiversidad, Universidad de Chile. Las Palmeras 3425, \~Nu\~noa, Santiago, Chile.}
\address{Departamento de Matem\'aticas, Universidad de Chile, Las Palmeras 3425, \~Nu\~noa, Santiago, Chile.}
\thanks{The first author is supported by Fondecyt Regular 1210733 (ANID-Chile)}

\begin{document}

\begin{abstract}
The metapopulation theory explores the population persistence in fragmented habitats by considering a balance between the extinction of local populations and recolonization of empty sites. In general, the extinction and colonization rates have been considered as constant parameters and the novelty of this paper is to assume that they are subject to deterministic variations. We noticed that an averaging approach 
proposed by C. Puccia and R. Levins can be adapted to construct the upper and lower averages of the difference between the extinction and colonization rates, whose sign is useful to determine
either the permanence or the extinction of the metapopulation. In fact, we use these averages to revisit the classical model introduced by R. Levins. From a mathematical perspective, these averages can be seen as Bohl exponents whereas the corresponding analysis is carried out by using tools of non autonomous dynamics. Last but not least, compared with the Levins model, the resulting dynamics of the time varying model shares the persistence/extinction scenario when the above stated upper and lower averages have the same sign but also raises open questions about metapopulation persistence in the case of the averages have different sign.

\end{abstract}

\maketitle
\section{Introduction}

The theory of metapopulations is a central framework in Ecology which emerged at the end of sixties \cite{Levins}, concomitantly with other fundamental theories such as   dynamic equilibrium theory and  island biogeography theory developed by MacArthur, Wilson and Simberloff \cite{MAW0},\cite{Simberloff}. Contrarily to the population models based on the difference between the processes of birth and death rates of a single species; the metapopulation research is focused on describing the proportion of the patches occupied by a species as a consequence of the dynamical balance between colonization and extinction \cite{HS}.



In general, metapopulation theory has worked with the  implicit or explicit assumption of a deterministic environment where the colonization and extinction processes are constant parameters. Nevertheless, in the real world,  birth rates, carrying capacities, competition coefficients  exhibit fluctuations in time \cite[p.109]{May}.  These criticisms have prompted to consider the time variability of the environment in several topics of theoretical ecology \cite{Fidino,Klausmeier,Kremer,Vesipa,White}. However, to the best of our knowledge, there are no metapopulation works on this vein with one exception which focuses on the almost periodic case \cite{Amster}. The seminal metapopulation model proposed by Richard Levins in 1969, defined the concept of metapopulation and described the dynamics by means of an ordinary differential equation, which included colonization and extinction rates as constants. 


An outcome of the Levins model is that the fraction of occupied patches converges exponentially to a static positive equilibrium when the colonization rate $c>0$ is higher that the extinction rate ($e > 0$);  in the reversed cases, the fraction of occupied patches is driven to zero. Although the Levins model has been generalized in multiple ways, we consider it relevant to take into account an additional one: instead of suppose constant colonization and extinction rates, we will assume that they are subject to deterministic temporal variability and are described by time--varying functions $c(t) > 0$ and $e(t) > 0$ respectively.


We stress that in \cite{Levins3} and \cite{Puccia-0,PC}, C. Puccia and R. Levins take into account the parameters variability mentioned above. They also enunciates a number of criticisms of equilibrium--based ecologic studies and proposes to consider the averages of ecologic variables. This idea is promising but, curiously, was not applied to its own metapopulation model.


In this work, we recall that there are not a univocal definition of average for continuous variables but the idea of C. Puccia and R. Levins can be reformulated in terms of; a formal definition will be given later; upper and lower averages of $e(t)-c(t)$, namely, the difference between extinction and colonization forces. These upper and lower averages coincides with the so called \textit{Bohl exponents}, which are in the core of the theory of nonautonomous dynamical systems, whose application in life sciences is on its preliminary stages \cite{Kloeden2}. Moreover, the outcome of our time varying version of the Levins model is that if the upper average of $e(t)-c(t)$ is negative then the proportion of patches is attracted either by a positive, bounded and continuous solution whereas is driven to zero if the lower average is positive.


The structure of the paper is as follows: In section 2 we revisit the Levins model and describe our generalization. In section 3 we describe the upper and lower average of the difference between extinction and colonization rate and state as the averages implies either the extinction of the metapopulation or its convergence towards a positive and continuous solution. We give a discussion in Section 4 and the Appendix provides formal mathematical proofs of our outcomes stated in section 3.

\section{The Levins metapopulation model and a generalization for time varying environments}

The theory of metapopulations assumes the existence of fragmented habitats and the concept of \textit{patch} plays
an essential role to describe this fragmentation. According to \cite{Chesson}, a patch is \textit{an area of suitable habitat for a particular 
species or particular collection of species, ideally bounded by unsuitable habitat or habitat with different physical properties}. Now, when considering a single species distributed in a fragmented habitat, the populations inhabiting at each specific patch are called \textit{local populations}. A \textit{metapopulation} is a collection of local populations in a region.

A classical or strict metapopulation \cite{Chesson} satisfies the fo\-llowing conditions: i) the \textit{patches} are  partially isolated between them and are able of autosustaining for several generations in absence of colonization for other local populations, ii) the local population extinction occurs in a time scale considering many generations, iii) migration between local populations leads to reestablishment of local populations following local extinction.

\subsection{The classical Levins model}
The concept of classical metapopulation has been introduced in 1969 by R. Levins \cite{Levins,Levins2}: A single species is assumed to inhabit a set of patches, which are assumed to be spatially homogeneous and having the same level of accessibility. The proportion of patches which 
are effectively occupied is denoted by $p\in [0,1]$ and its variation in time is go\-ver\-ned by a dynamics between two opposing forces: colonization and extinction, which can be described in a very general way as follows  \cite[p.391]{Case}:
\begin{equation}
\label{levins-00}
p'=\textnormal{[gains from colonization]}-\textnormal{[losses due to extinction]},
\end{equation}
where the gains from colonization are described by the product
of the intrinsic patch colonization rate $c>0$ with the fractions of occupied patches $p$
and unoccupied patches $1-p$. On the other hand, the extinction per--capita rate
is assumed to be constant, namely $e>0$. In consequence, the classical Levins model is described
by the differential equation:
\begin{equation}
\label{levins-0}
p'=cp(1-p)-ep.
\end{equation}

If $p(0)\in (0,1]$, the evolution of the solutions of (\ref{levins-0}) in the future is uniquely
determined by the sign of $e-c$ and is described by the following behavior:
\begin{equation}
\label{NALM}
\lim\limits_{t\to +\infty}p(t)=\left\{
\begin{array}{rcl}
\displaystyle 0  & \textnormal{if}&   c \leq e \\\\
\displaystyle p^{*}:=1-\frac{e}{c}  &\textnormal{if}  & e<c. 
\end{array}\right.
\end{equation}

\begin{remark}
\label{outcome}
When defining $\mathcal{D}:=e-c$, namely, the difference between extinction and colonization, the above limit can be seen as a persistence/extinction scenario characterized by an attractive equilibrium, which is determined by the excluding 
conditions $\mathcal{D}<0$, $\mathcal{D}>0$ and  $\mathcal{D}=0$ that represent either the dominance of colonization over extinction or its opposites. In addition:
\begin{enumerate}
\item[i)] when $\mathcal{D}<0$, the fraction of occupied patches converges
towards the constant $p^{*}=1-(e/c)$ at exponential rate, 
\item[ii)] when $\mathcal{D}>0$, the fraction of occupied patches is driven to zero at exponential rate,
\item[iii)] when  $\mathcal{D}=0$, the fraction of occupied patches is driven to zero but not necessarily at an exponential rate. 
\end{enumerate}
\end{remark}

\begin{remark}
\label{outcome3}
In \cite{Levins2}, Levins considers the extinction rate as a random variable of the frequency of a gene in the population
whose average is $\overline{e}$ and its variance is $\sigma_{e}^{2}$. If $c>\overline{e}+\sigma_{e}^{2}$ it is proved the existence of
a modal distribution with a peak at $p=1-(\overline{e}+\sigma_{e}^{2})/c$. This idea has been revisited in several works by considering more general stochastic extinction rates. We refer the reader to \cite{Nothaab} for recent results.
\end{remark}

\subsection{Novelty of this work: A generalization by considering time--varying per--patch colonization rates}
In \cite[p.765]{Levins3}, R. Levins enumerates several critical remarks to the coexistence theory based on equilibria and a key
one states that \textit{equilibrium theory is simply innapropriate when the environment is varying, for example seasonally}. This critical stance towards the notion of equilibrium is shared by Etienne and Nagelkerke \cite{Etienne,Etienne2} which compares  the classic deterministic Levins model (\ref{levins-0}) with an analogous stochastic model analogous carried out by Frank $\&$ Nagelkerke \cite{Frank} and Ovaskainen \cite{Ovaskainen}.  Nevertheless, to the best of our knowledge, the Levins model with deterministic time varying colonization and extinction rates has not been studied previously.

In this context, we will generalize the Levins model (\ref{levins-0}) by considering time varying colonization and extinction rate such that the equation (\ref{levins-0}) becomes a time--varying or nonautonomous differential equation:
\begin{equation}
\label{Levins1}
p'=c(t)p(1-p)-e(t)p,
\end{equation}
where $c\colon [0,+\infty)\to [0,+\infty)$ and $e\colon [0,+\infty) \to [0,+\infty)$
are bounded and continuous functions. Moreover, we will assume that colonization rate is strictly positive, that is,
there exists $c_{*}>0$ such that:
\begin{equation}
\label{colonization-rest}
\inf\limits_{t\geq 0}c(t)\geq c_{*}>0.
\end{equation}

The goal of this article is to study the evolution of the solutions of the time varying
Levin's model (\ref{Levins1}) by generalizing the above mentioned threshold condition.
A first attempt to study (\ref{Levins1}) has been done by the second author \cite{Amster} which considers almost periodic colonization and extinction rates. As it has been 
stated in \cite{Kloeden2}, several techniques and methods from time varying ordinary differential equations can be employed to study models arising from life sciences and this article can be seen on this trend.

\section{Averages for the colonization/extinction rates and persistence/extinction outcomes}

In order to study the equation (\ref{Levins1}) and the properties of its solutions, it will be useful to define the difference between the extinction and colonization rates at time $t$ by the function $t\mapsto \mathcal{D}(t)$ described by:
\begin{equation}
\label{D}
\mathcal{D}(t):=e(t)-c(t) \quad \textnormal{for any $t\geq 0$}.
\end{equation}

Note that the sign of $\mathcal{D}(t)$ describes if either the colonization or the extinction are dominant forces at time $t$. From this perspective, if $\mathcal{D}(t)$ has infinite changes of sign then there will be infinite intervals where the colonization is dominated by the extinction and vice versa and, consequently, a persistence/extinction scenario cannot be directly deduced as in the classical case described by the equation (\ref{levins-0}). 

In this article we point out that, instead of considering the sign of $\mathcal{D}(t)$, a more beneficial approach is to consider mathematical objects which emulates the \textit{average} of the extinction/colonization balance described by $t\mapsto \mathcal{D}(t)$, which allows us to raise two fundamental questions:
there will be always a fraction of occupied patches if the average of $\mathcal{D}(\cdot)=e(\cdot)-c(t)$ is negative?. The metapopulation will be driven to the extinction if the average of $e(\cdot)-c(\cdot)$ is positive?.

The preceding questions are a natural generalization of the persistence/extinction scenario described by the equation (\ref{NALM}). However, despite its apparent simplicity, we face a problem: there are no univocal definition 
for the average of a continuous function and there exist several possible ways to define it. 

\subsection{Averaging approaches}
In \cite{Puccia-0} and \cite[Ch.7]{PC}, C.J. Puccia and R. Levins proposed a formal definition for the average or \textit{expected value} of a bounded continous function, which applied to the difference between extinction and colonization $\mathcal{D}(\cdot)$, would be:
\begin{equation}
\label{APL}
E_{t}(\mathcal{D})=\frac{1}{t}\int_{0}^{t}\mathcal{D}(\tau)\,d\tau \quad \textnormal{for any $t>0$}.
\end{equation}

The usefulness of the expected value is nicely described by several examples and applications
developed in \cite{Puccia-0} and \cite{PC}. Surprisingly, to the best of our knowledge, 
the expected value approach seems to have not been used to the study of metapopulations. In 
addition, a careful reading of the Puccia and Levins work shows that they have the underlying assumption that
the identity (\ref{APL}) must be considered with large enough values of $t$. In fact, this assumption
was explicit on \cite{Levins3} and implicit in \cite[p.195]{PC} where the authors stated that \textit{"$\ldots$ for the rest of this section we will take average in the limit and drop the subscript"}. 

We point out that the expected value (\ref{APL}) assumes an initial time
$t_{0}=0$ but the dynamics of the time varying systems can be dependent
of the initial time. In this context a natural generalization adapted for any initial time $t_{0}=s$ could be
\begin{equation}
\label{APL2}
E_{t,s}(\mathcal{D})=\frac{1}{t-s}\int_{s}^{t}\mathcal{D}(\tau)\,d\tau \quad \textnormal{for any $t>s\geq 0$}.
\end{equation}

The expected value $E_{t}(\mathcal{D})$ and its generalization $E_{t,s}(\mathcal{D})$ are related the solutions of the linear time varying diferential equation:
\begin{equation}
\label{lin-v}
\dot{v}=\mathcal{D}(t)v \quad \textnormal{or equivalently} \quad \dot{v}=[e(t)-c(t)]v \quad \textnormal{for any $t\geq 0$},
\end{equation}
whose transition operator is:
\begin{equation}
\label{MTD}
\Phi(t,s)=\exp\left(\int_{s}^{t}\mathcal{D}(\tau)\,d\tau\right) \quad \textnormal{for any $t,s\geq 0$}.
\end{equation}

The solution $t\mapsto v(t)$ of (\ref{lin-v}) passing through $v_{0}$ at time $t_{0}$ is usually denoted as follows: 
$$
v(t)=\Phi(t,t_{0})v_{0} \quad \textnormal{for any $t,t_{0}\geq 0$},
$$
then, when considering $t_{0}\geq 0$, the solutions of (\ref{lin-v}) can be written as 
\begin{equation}
\label{LevAv}
v(t)=e^{E_{t}(\mathcal{D})t}v(0) \quad \textnormal{or} \quad
v(t)=e^{E_{t,t_{0}}(\mathcal{D})(t-t_{0})}v(t_{0}) \quad \textnormal{for any $t>t_{0}\geq 0$}.
\end{equation}

The solutions of (\ref{lin-v}) describes the temporal evolution of the opposing forces of colonization and extinction. In particular, notice that the solutions (\ref{LevAv}) are dependent of the sign of $E_{t}(\mathcal{D})$ and $E_{t,s}(\mathcal{D})$: decreases when the expected value is negative and increases when is positive.

A strong shortcoming of the Puccia and Levins approach is that the expected value described by eq. (\ref{APL}) (\textit{resp}. eq. (\ref{APL2})) could not be always well defined when $t$ (\textit{resp}. $t-s$) tends to infinity. Nevertheless, noteworthy mathematical tools have been developed in order to study the solutions of the linear equation (\ref{lin-v}) which generalize the notion of expected value, the first one is given by the \textit{lower Lyapunov exponents} and \textit{upper Lyapunov exponents} of $\mathcal{D}(\cdot)$ described by the limits:
\begin{equation}
\label{Lyapunov}
\lambda^{-}(\mathcal{D})=\liminf\limits_{t\to +\infty}\frac{1}{t}\int_{0}^{t}\mathcal{D}(\tau)\,d\tau \quad \textnormal{and} \quad
\lambda^{+}(\mathcal{D})=\limsup\limits_{t\to +\infty}\frac{1}{t}\int_{0}^{t}\mathcal{D}(\tau)\,d\tau.
\end{equation}

Note that Lyapunov's exponents agree with the idea of Puccia and Levins insofar as they consider the expected value of $\mathcal{D}(t)$ by making $t\to +\infty$. However, since such a limit may not exist, upper and lower limits must be considered.

A second approach to the expected value of $\mathcal{D}(\cdot)$ is given by the \textit{upper Bohl exponent} and the \textit{lower Bohl exponent} which, applied to the difference between extinction and colonization $\mathcal{D}(\cdot)$, can be denoted as
\begin{equation}
\label{Bohl}
\beta^{-}(\mathcal{D})=\liminf\limits_{t-s\to +\infty}\frac{1}{t-s}\int_{s}^{t}\mathcal{D}(r)\,dr
\end{equation}
and
\begin{equation}
\label{Bohl-2}
\beta^{+}(\mathcal{D})=\limsup\limits_{t-s\to +\infty}\frac{1}{t-s}\int_{s}^{t}\mathcal{D}(r)\,dr.
\end{equation}

Similarly as the Lyapunov exponents, we can see that Bohl's exponents also allow us to revisit the generalization the expected value of Puccia and Levins described by $E_{t,s}(\mathcal{D})$ which is adapted for any initial time. In fact:
\begin{itemize}
\item[$\bullet$] If $\beta^{+}(\mathcal{D})<0$ it follows that that, notwithstanding any variations in the sign of the function $t\mapsto \mathcal{D}(t)=e(t)-c(t)$, the difference between colonization and extinction forces
is negative from an upper average perspective, that is, the colonization forces are dominating in the metapopulation.
\item[$\bullet$] If $\beta^{-}(\mathcal{D})>0$ it follows that the difference between colonization and extinction forces
is positive from an upper average perspective, that is, the extinction forces are dominating in the metapopulation.
\end{itemize}

The limits (\ref{Bohl}) and (\ref{Bohl-2}) will be respectively called 
as \textit{lower average} and \textit{upper average} of $\mathcal{D}(\cdot)$. As we have said, Lyapunov and Bohl exponents are essential tools to study linear systems as (\ref{lin-v}) and we would like to refer to \cite{Barreira,Kloeden} for a deep treatment of Lyapunov exponents
and \cite{APR,barabanov2001,Bohl,DK} for additional details about Bohl exponents. In particular, see \textit{e.g.} \cite{APR}, it can be proved that
\begin{displaymath}
\beta^{-}(\mathcal{D})\leq \lambda^{-}(\mathcal{D})\leq \lambda^{+}(\mathcal{D}) \leq \beta^{+}(\mathcal{D}),    
\end{displaymath}
and the identity $\beta^{-}(\mathcal{D})=\beta^{+}(\mathcal{D})$ is verified in special cases, for example when $\mathcal{D}(\cdot)$ is constant, periodic or almost periodic. 

Despite our initial criticism to the notion of expected value $E_{t}(\mathcal{D})$ introduced 
by Puccia and Levins, we have to say that this is an idea that, behind a false naivety, hides great depth.
In particular, the generalization of the expected values of $\mathcal{D}(\cdot)$ by Bohl exponents will allow to study
the metapopulation described by (\ref{Levins1}) which is because the sign of these exponents determines the exponential stability/instability of the solutions of (\ref{lin-v}) and (\ref{Levins1}).

\subsection{Bohl exponents and persistence/extinction results} 
The Bohl exponents $\beta^{+}(\mathcal{D})$ and $\beta^{-}(\mathcal{D})$ allow a partial emulation of the classic Levins result
described by the equation (\ref{NALM}), that is,
a persistence/extinction scenario characterized by an exponentially attractive solution. In this context, the main results
of this note are:

\medskip

\subsubsection{Persistence scenario:} If $\beta^{+}(\mathcal{D})<0$, or equivalently, the colonization dominates over the extinction from an upper average point of view then the fraction of the occupied patches, independently of the initial condition, is described by the function:
\begin{equation}
\label{BS}
p(t)=\left(\int_{0}^{t}\Phi(t,s)c(s)\,ds+o(1)\right)^{-1},
\end{equation}
where $\textit{o}(1)$ describes a function exponentially convergent to zero when $t\to +\infty$. 

Moreover, the expected value $E_{t}(p)$
in the sense of Puccia and Levins is given by:
\begin{equation}
\label{EVP}
E_{t}(p)=1-\frac{E_{t}(e)}{E_{t}(c)}+\frac{\text{Cov}(c,p)}{E_{t}(c)},
\end{equation}
and we can see that the above expression is reminiscent to (\ref{NALM}) when the colonization forces
dominates over the extinction ones, that is when $E_{t}(c)>E_{t}(e)$. The symbol $\text{Cov}(c,p)$ denotes the covariance between the colonization rate and the fraction of occupied patches, which is described by
\begin{displaymath}
\text{Cov}(c,p)=E_{t}(cp)-E_{t}(c)E_{t}(p).    
\end{displaymath}

\subsubsection{Extinction scenario:} If $\beta^{-}(\mathcal{D})>0$, that is, if the colonization is dominated by the extinction
from a lower average point of view, then the fraction of occupied patches converges exponentially towards zero when $t\to +\infty$.

 \begin{figure}
                    \centering
                    \includegraphics[scale=0.45]{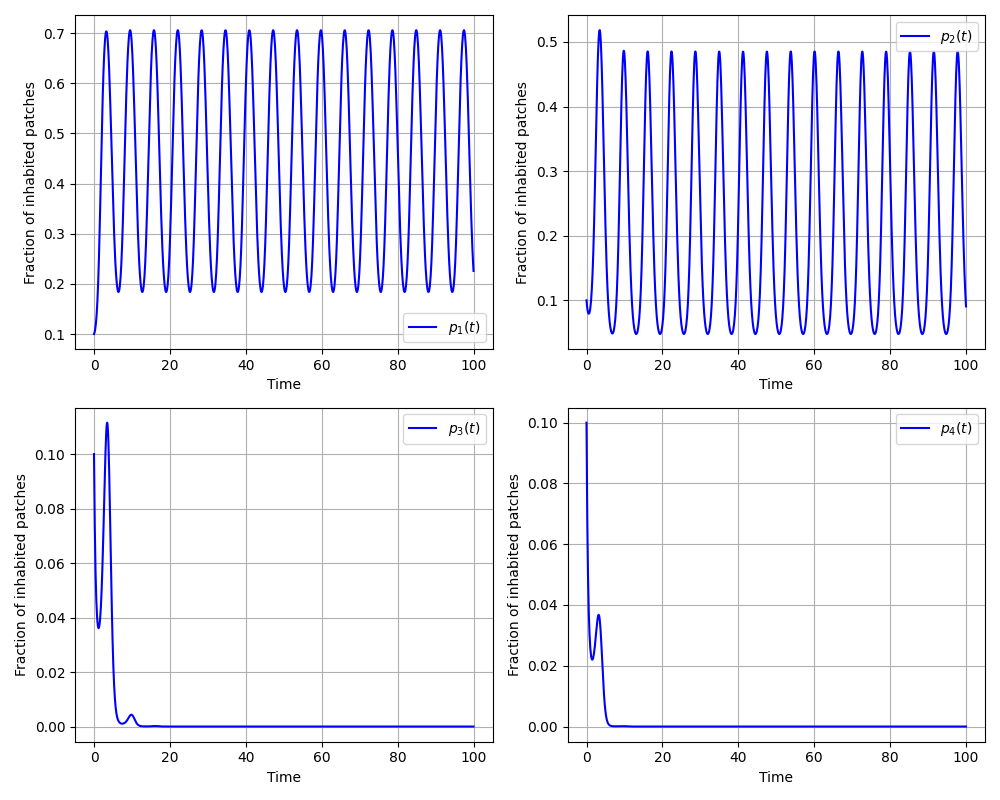}
                    \caption{Fraction of occupied patches by considering the colonization and extinction rates: 
                    $c(t)=\alpha+\sin(t) \quad \textnormal{and} \quad e(t)=\beta+\cos(t)$.
                    The images of the first row are obtained by considering the parameters $(\alpha,\beta)=(3.5,2)$ and $(\alpha,\beta)=(2.5,2)$, which leads respectively to the averages
                    $\beta^{+}(\mathcal{D})=\beta^{-}(\mathcal{D})=-1.5$ and $\beta^{+}(\mathcal{D})=\beta^{-}(\mathcal{D})=-0.5$, the persistence of the metapopulation is observed. The images of the second row are obtained by considering the parameters $(\alpha,\beta)=(1.5,2)$ and $(\alpha,\beta)=(1.1,2)$, leading to $\beta^{+}(\mathcal{D})=\beta^{-}(\mathcal{D})=0.5$ and $\beta^{+}(\mathcal{D})=\beta^{-}(\mathcal{D})=0.9$, the extinction of the metapopulation is observed.}
                    \label{fig1}
                \end{figure}

\subsection{Partial outcomes and open problems}
The scenarios complementary to those above mentioned are:
\begin{equation}
\label{abs-dico}
\beta^{-}(\mathcal{D})\leq 0 < \beta^{+}(\mathcal{D}) \quad \textnormal{and/or} \quad \beta^{-}(\mathcal{D})<0\leq \beta^{+}(\mathcal{D}),
\end{equation}
or equivalently 
$0\in [\beta^{-}(\mathcal{D}),\beta^{+}(\mathcal{D})]$
and, there are no strict dominance neither of colonization forces nor
extinction ones, that is, the variations of the sign of the function $\mathcal{D}(t)=e(t)-c(t)$ describing
the difference between extinction and colonization forces are such that the following simultaneous properties arise:
\begin{enumerate}
 \item[$\bullet$] From an upper average perspective: the colonization forces dominate 
 over the extinction ones,
 \item[$\bullet$] From a lower average perspective: the extinction forces 
 dominate over the colonization ones.
\end{enumerate}

Under these conditions, in the following sections we will see how we only can prove that the fraction of occupied patches 
satisfies $\liminf\limits_{t\to +\infty}p(t)=0$ whereas the properties of the upper limit have not yet been studied in depth.

\section{Discussion}

\subsection{Similarities with the classic case} 
The results described in subsection 3.1 generalize those from the Levins model since are reduced to classical ones stated by Eq.(\ref{NALM}) and Remark \ref{outcome} when the extinction and colonization rates are positive constants constants $e$ and $c$ verifying $d\neq c$. In adddition, from a qualitative perspective.  since there still exists a persistence/extinction scenario, which is now determined by 
the upper and lower averages of the function $t\mapsto \mathcal{D}(t):=e(t)-c(t)$ as follows:

\begin{enumerate}
\item[$\bullet$] When $\beta^{+}(e-c)<0$ there is a strict dominance
of the colonization forces over the extinction ones. In this case,
the fraction of occupied patches $t\mapsto p(t)$ is persistent
and is described by the asymptotic formula (\ref{BS}) where $o(1)$ converges exponentially to zero. This property becomes the statement i) from Remark \ref{outcome} in the constant case.
\item[$\bullet$] when $\beta^{-}(e-c)>0$ there is a strict dominance of the extinction forces over the colonization ones. In this case, the fraction of occupied patches is driven to zero at exponential rate. This property coincides with the statement ii) from Remark \ref{outcome} in the constant case.
\end{enumerate}

\subsection{Differences with the classic case} 
The results described in subsection 3.2 only partially recover those from the classical Levins model
stated by Eq.(\ref{NALM}) and Remark \ref{outcome} when the colonization forces are identical
to the extinction ones.

The methods and ideas employed in the previous cases only allow us to prove that $\liminf\limits_{t\to +\infty}p(t)=0$
but we cannot prove a similar result for the upper limit. 

The current state of art in this case is reduced to the following partial results: 

\begin{enumerate}
\item[$\bullet$] In \cite[Lemma 6.3]{Amster} we proved that, if $t\mapsto c(t)$ and $t\mapsto e(t)$ are
$T$--periodic functions, it follows that  $\limsup\limits_{t\to +\infty}p(t)=0$. 
\item[$\bullet$] In the Appendix of this work we construct a function
$t\mapsto d(t)$ such that the scenario $\limsup\limits_{t\to +\infty}p(t)>0$ is theoretically possible and 
would imply a weak persistence of the metapopulation in the sense
stated by H. Thieme in \cite[pp.175--176]{Thieme}. 
\end{enumerate}

Despite the above partial results, to determine
sufficient conditions ensuring either extinction or weak persistence
of the metapopulation remains a challenging problem.

\subsection{Mathematical remarks}
The linear time varying equation (\ref{lin-v})
played an essential role in the proof of our results. As we stated before, its non trivial solutions  describe the opposing forces of colonization and extinction. In this context, the strict dominance of colonization over extinction $\beta^{+}(e-c)<0$ and vice versa $\beta^{-}(e-c)>0$, is equivalent to the property of exponential dichotomy of (\ref{lin-v})
whose consequences will allow us to carry out a careful mathematical
analysis of the solutions of (\ref{Levins1}).

We believe that the methods and ideas employed in this article could also be used in more general metapopulations models as those described in the seccion 2.

\subsection{Extension of Levins model}
There is a myriad of possible ways to ge\-ne\-ralize the Levins model and we refer to \cite[Table 1]{Marquet}, \cite[Ch.4]{Hanski99}\cite[p.82]{OH} for details. In this subsection, we will be interested in approaches based on scalar differential equations described as follows:
\begin{equation}
\label{gene1}
p'=c(t)\mathcal{C}(p)-e(t)\mathcal{E}(p)
\end{equation}
where $\mathcal{C}(p)$ and $\mathcal{E}(p)$ respectively describes more general forms of gains from colonization and losses due to extinction. Notice that the equation (\ref{Levins1})
is encompassed by (\ref{gene1}) by considering the particular cases $\mathcal{C}_{L}(p)=p(1-p)$
and $\mathcal{E}_{L}(p)=p$. We will distinguish two specific behaviors: the \textit{rescue effect} and the \textit{propagule rain}.

\medskip

The rescue effect has been introduced by Brown $\&$ Kodric--Brown \cite{Brown} which suggested that colonization and extinction are not independent and, in this context, \textit{an island or patch that receives a high rate of colonization will also have a hight rate of immigration when it is already occupied $\cdots$ they called this the \textbf{rescue effect} since a population is rescued from extinction by continual immigration} \cite[p.401]{Case}.

The rescue effect has been revisited in 1982 by I. Hanski in \cite{Hanski82} by considering
a quadratic expression for the losses due to extinction:
\begin{displaymath}
\mathcal{E}(p)=p(1-p),    
\end{displaymath}
that is, the extinction losses are increasing only for small values of $p$ whereas are decreasing for values of $p$ close to $1$. In particular, the rescue effect is displayed when fractions of patches closed to $p=1$ leads to $\mathcal{E}(p)$ closed to $0$.

The \textit{propagule rain} has been studied by Gotelli in \cite{Gotelli} and assumes that
the colonization rate depend on the fraction of unoccupied patches:
\begin{displaymath}
\mathcal{C}(p)=1-p.    
\end{displaymath}

In \cite[Fig.2]{Gotelli}, Gotelli proposes a dichotomic  classification
$\mathcal{E}(p)$ v/s and $\mathcal{E}_{L}(p)$ together with $\mathcal{C}(p)$ v/s and $\mathcal{C}_{L}(p)$. This table has been revisited by the first author in \cite[Table 6]{Bustamante}: 
\begin{table}[h]
   \centering
    \begin{tabular}{|c|c|c|}
    \hline
  & \textnormal{\small{$\mathcal{E}(p)=p(1-p)$}} & \textnormal{\small{$\mathcal{E}_{L}(p)=p$}} \\      
\hline

\textnormal{\small{$\mathcal{C}(p)=(1-p)$}}   &  \textnormal{Case I)}   &  \textnormal{Case III)}    \\
\hline 
\textnormal{\small{$\mathcal{C}_{L}(p)=p(1-p)$}} & \textnormal{Case II)}  &  \textnormal{Levins Model}    \\ 

\hline
    \end{tabular}
    \caption{Extensions of the classical Levins model. The first co\-lumn considers the rescue effect whereas the first row considers the propagule rain.}
    \label{tab_0}
\end{table}

The cases I,II and III from the above table can studied similarly as we described in section 3, we can summarize 
our results as follows:

\medskip

\subsubsection{Study of case I:} The model with propagule rain
and rescue effect has been studied by Gotelli in \cite{Gotelli}. Its generalized version is given by:
\begin{equation}
\label{PR}
p'=c(t)(1-p)-e(t)p(1-p).
\end{equation}

\noindent $\bullet$ If $\beta^{+}(e-c)<0$ it can be proved that the fraction
of occupied patches converges exponentially towards $1$

\noindent $\bullet$ If $\beta^{-}(e-c)>0$ it can be proved that any solution
$t\mapsto p(t)$ verifies
\begin{displaymath}
p(t)=1-\left(\int_{0}^{t}\Phi^{-1}(t,s)e(s)\,ds+o(1)\right)^{-1}.    
\end{displaymath}

We point out that if $c(t)=c$ and $e(t)=e$, the above expression coincides
with the attractiveness of the equilibrium $c/e$ when $c<e$. Notice that, in this case, when the colonization forces are dominant from an upper average pont of view, all the patches will be occupied. In addition, there are no extinction of the metapopulation when the extinction forces are dominatinf from a lower average point of view. 

\subsubsection{Study of case II:} The model with rescue effect and $\mathcal{C}_{L}(p)=p(1-p)$ was introduced
by Hanski in \cite{Hanski82} with constant colonization/extinction rates. When considering time varying colonization and extinction rates, the model becomes:
\begin{equation}
\label{RE}
p'=c(t)p(1-p)-e(t)p(1-p).
\end{equation}

The study of the solutions of (\ref{RE}) can be addressed by using similar methods as those employed to study 
(\ref{Levins1}). In fact, it can be proved that:
\medskip

\noindent $\bullet$ If $\beta^{+}(e-c)<0$ then any solution
$t\mapsto p(t)$ with $p(0)\in (0,1]$ verifies
$$
p(t)=\left(-\int_{0}^{t}\Phi(t,s)d(s)\,ds+o(1)\right)^{-1}=\frac{1}{1+o(1)},
$$
where the last identity follows from $\frac{d}{ds}\Phi(t,s)=-\Phi(t,s)d(s)$
and the asymptotic term is $o(1)=\Phi(t,s)[p^{-1}(0)-1]$.

\noindent $\bullet$ If $\beta^{-}(e-c)>0$ then any solution of 
(\ref{RE}) converges exponentially towards zero.

Similarly as in the previous case, we can see that in presence of rescue effect, the persistence of the metapopulations is achieved with the totality of the patches occupied.

\subsubsection{Study of case III:} The model with propagule rain $\mathcal{C}(p)=1-p$ and without rescue effect 
can also be seen as a single species version of the mainland--island model of Mac-Arthur and Wilson \cite{MAW}. When considering time varying colonization and extinction rates it becomes:
\begin{equation}
\label{MI}
p'=c(t)(1-p)-e(t)p,
\end{equation}
and it is straightforward to deduce that any solution $t\mapsto p(t)$ of (\ref{MI}) 
verifies: 
$$
p(t)=\int_{0}^{t}\Psi(t,s)c(s)\,ds +o(1)  \quad \textnormal{with $\Psi(t,s)=\exp\left(-\int_{s}^{t}[c(r)+e(r)]\,dr\right)$}.
$$

\medskip

Finally, we also stress 
the contribution of Hanski and Gyllenberg \cite{HG}, where the \textit{phenomenological model}:
\begin{displaymath}
p'=\{p/[b+(1-b)p]\}\{c(t)(1-p)-e(t)[b^{2}+(1-b^{2})p]\}    \quad \textnormal{with $b\in [0,1]$}
\end{displaymath}
is introduced. Note that the limit case $b=1$ coincides with the Levins model (\ref{levins-0}) 
while the limit case $b=0$ coincides with the mainland--island model (\ref{MI}). Nevertheless, we point out that our methods and ideas cannot be applied to this model and its study remains an open problems.

\appendix

\section{Persistence/extinction average conditions for metapopulations: Mathematical proofs}

In order to study the behavior of the solutions $t\mapsto v(t)$  of (\ref{lin-v}) with
the Bohl exponents $\beta^{+}(\mathcal{D})$ and $\beta^{-}(\mathcal{D})$, we need to consider the shifted equation associated to (\ref{lin-v}): 
\begin{equation}
\label{shifted}
x'=[\mathcal{D}(t)-\lambda]x,
\end{equation}
where $\lambda \in \mathbb{R}$. Notice that its transition operator is
$\Phi(t,s)e^{-\lambda(t-s)}$. The following definition will be fundamental:

\begin{definition}
\label{dichotomie}
The shifted system \eqref{shifted} has an exponential dichotomy on $[0,+\infty)$ if there exist constants
$K,L\geq 1$ and $\alpha,\eta>0$ such that  either
\begin{equation}
\label{ED+}
\Phi(\tau,s)e^{-\lambda(\tau-s)} \leq Ke^{-\alpha(\tau-s)} \quad \textnormal{for any $\tau \geq s$ with $\tau,s\geq 0$},   
\end{equation}
or
\begin{equation}
\label{ED-}
\Phi(\tau,s)e^{-\lambda(\tau-s)} \leq Le^{-\eta(s-\tau)} \quad \textnormal{for any $s\geq \tau$ with $\tau,s\geq 0$}.   
\end{equation}
\end{definition}

Last but not least, we need to introduce the following concept: 
\begin{definition}
\label{espectro}
The exponential dichotomy spectrum of the equation \eqref{lin-v} is defined by the following set:
\begin{displaymath}
\Sigma(\mathcal{D})=\left\{\lambda \in \mathbb{R}\colon x'=[\mathcal{D}(t)-\lambda]x \quad \textnormal{has not an exponential dichotomy on $[0,+\infty)$}\right\}.
\end{displaymath}
\end{definition}

The above spectrum has a simpler characterization in this scalar case:
\begin{proposition} 
\label{SPEC}
The exponential dichotomy spectrum of \eqref{lin-v} is a closed interval
\begin{displaymath}
\Sigma(\mathcal{D})=[\beta^{-}(\mathcal{D}),\beta^{+}(\mathcal{D})],
\end{displaymath}
and its complement is composed by the spectral gaps $(-\infty,\beta^{-}(\mathcal{D}))$
and $(\beta^{+}(\mathcal{D}),+\infty)$.
\begin{itemize}
\item The case \eqref{ED+} holds for $\mathcal{D}(\cdot)-\lambda$ if and only if $\lambda \in (\beta^{+}(\mathcal{D}),+\infty)$,
\item The case \eqref{ED-} holds for $\mathcal{D}(\cdot)-\lambda$ if and only if $\lambda \in (-\infty,\beta^{-}(\mathcal{D}))$.
 \end{itemize}
\end{proposition}

We refer the reader to \cite{Amster} for a detailed proof. By considering $\lambda=0$, a direct consequence
follows:
\begin{remark}
\label{consecuencia}
The linear system has an exponential dichotomy on $[0,+\infty)$ if and only if $0\notin [\beta^{-}(\mathcal{D}),\beta^{+}(\mathcal{D})]$.
\end{remark}

Moreover, the following result shows that the Bohl exponents will allow to characterize the behavior of the solutions (\ref{lin-v}) for bigger values of $t$:
\begin{proposition}
\label{PROP}
Let $t\mapsto x(t)$ be a solution of \eqref{lin-v} passing trough $x_{0}\neq 0$ at time
$t_{0}\geq 0$, then we have:

\medskip

\noindent \textnormal{a)} If $\beta^{+}(\mathcal{D})=\beta^{+}(e-c)<0$, then $\lim\limits_{t\to +\infty}x(t)=0$ and the convergence is 
uniformly exponential, that is, there exists constants $K\geq 1$ and $\alpha>0$ such that:
$$
|x(t)| \leq K|x_{0}|e^{-\alpha(t-t_{0})} \quad \textnormal{\textit{for all} $t\geq t_{0}\geq 0$},
$$

\noindent \textnormal{b)} If $\beta^{-}(\mathcal{D})=\beta^{-}(e-c)>0$, then $\lim\limits_{t\to +\infty}|x(t)|=+\infty$ and the divergence is uniformly exponential, that is, there exists constants $L\geq 1$ and $\eta >0$ such that:
$$
|x_{0}| \leq L|x(t)|e^{-\eta(t-t_{0})} \quad \textnormal{or} \quad |x_{0}|e^{\eta(t-t_{0})}\leq |x(t)| \quad \textnormal{for all $0\leq t_{0}\leq t$}.
$$
\end{proposition}

\begin{proof}
In order to prove the statement a), note that if $\beta^{+}(d)<0$, the characterization of the exponential dichotomy spectrum implies that $\Sigma(d)\subset (-\infty,0)$. Then, we can consider (\ref{ED+}) with $\lambda=0$ and $(\tau,s)=(t,t_{0})$ to deduce that:
$$
\Phi(t,t_{0})|x(t_{0})| \leq K|x(t_{0})|e^{-\alpha(t-t_{0})} \quad \textnormal{for any $t\geq t_{0}\geq 0$}.   
$$

As $x(t)=\Phi(t,t_{0})x(t_{0})$, the above inequality is equivalent to
$$
|x(t)|\leq  K|x(t_{0})|e^{-\alpha(t-t_{0})} \quad \textnormal{for any $t\geq t_{0}\geq 0$},
$$
and the statement a) follows.

To prove the statement b), let us observe that if $\beta^{-}(d)>0$, the characterization of the exponential dichotomy
spectrum implies that $\Sigma(d)\subset (0,+\infty)$. In consequence, we can consider (\ref{ED-}) with $\lambda=0$  and $(\tau,s)=(t_{0},t)$, which allow us to deduce that
$$
\Phi(t_{0},t)|x(t)| \leq  Le^{-\eta(t-t_{0})}|x(t)| \quad \textnormal{for any $t\geq t_{0}\geq 0$}.
$$

As $x(t_{0})=\Phi(t_{0},t)x(t)$, the above inequality implies that
$$
|x(t_{0})|\leq  L|x(t)|e^{-\eta(t-t_{0})} \quad \textnormal{for any $t\geq t_{0}\geq 0$},
$$
and the statement b) follows.

\end{proof}

The above result is mentioned in the literature but we include it in order to make this note
the most self contained as possible.

\begin{theorem}
\label{T1}
Let us consider the generalized Levins model described by \eqref{Levins1} where 
$t\mapsto c(t)$ and $t\mapsto e(t)$ are bounded, continuous and positive maps for any $t\in \mathbb{R}$. If $p(0)>0$ is the fraction of the occupied patches at time zero, then the fraction $t\mapsto p(t)$ of occupied patches at time $t\geq 0$ is described by:
\begin{equation}
\label{limits}
p(t)= \displaystyle\frac{1}{\displaystyle \Phi(t,0)p^{-1}(0)+\int_{0}^{t}\Phi(t,s)c(s)\,ds}.
\end{equation}
Moreover its evolution in time will be determined by the relation between the opposing forces of colonization and extinction as follows:
as follows:
\begin{itemize}
\item[i)] If $\beta^{+}(\mathcal{D})=\beta^{+}(e-c)<0$, the fraction $t\mapsto p(t)$ of occupied patches verifies
\begin{equation}
\label{AB}
p(t)=\left(\int_{0}^{t}\Phi(t,s)c(s)\,ds+o(1)\right)^{-1},
\end{equation}
where $\textit{o(1)}$ is a function exponentially convergent to zero.
\item[ii)] If $\beta^{-}(\mathcal{D})=\beta^{-}(e-c)>0$, there exists $L\geq 1$ and $\eta>0$ such that the fraction of occupied patches $t\mapsto p(t)$ is driven exponentially to zero at rate
\begin{equation}
\label{limit2}
p(t)\leq Le^{-\eta t}  \quad \textnormal{for any $t\geq 0$}.
\end{equation}
\end{itemize}
\end{theorem}

\begin{proof}
Notice that the differential equation (\ref{Levins1}) is of Bernoulli's type, then the change of variables $u=1/p$ combined
with  $\mathcal{D}(t)=e(t)-c(t)$ transforms (\ref{Levins1}) in the inhomogenous equation:
\begin{equation}
\label{Levins2}
u'=\mathcal{D}(t)u+c(t).
\end{equation}

Notice that the linear equation associated to (\ref{Levins2}) is given by (\ref{lin-v}). By variation of parameters, it is direct to verify that
\begin{equation}
\label{solucion}
u(t)=\Phi(t,0)u(0)+\int_{0}^{t}\Phi(t,s)c(s)\,ds
\end{equation}
and the identity (\ref{limits}) follows due to $u(t)=1/p(t)$.

Firstly, we assume that $\beta^{+}(\mathcal{D})<0$ and let $t\mapsto u(t)$ be any solution of (\ref{Levins2}).
Now, by using the statement a) from Proposition \ref{PROP}, we have the existence of constants $K\geq 1$ and $\alpha>0$ such that any non trivial solution $t\mapsto v(t)=\Phi(t,0)p^{-1}(0)$ of (\ref{lin-v}) verifies:
$$
|v(t)|=\Phi(t,0)|v(0)|\leq Ke^{-\alpha t}|v(0)| \,\,\, \textnormal{or equivalently $\Phi(t,0)\leq Ke^{-\alpha t}$} \,\,\,\textnormal{for any $t\geq 0$},
$$
then we conclude easily that $
\Phi(t,0)p^{-1}(0)=\textit{o}(1)$ and the convergence is exponential when $t\to +\infty$, which leads to
\begin{displaymath}
u(t)=\int_{0}^{t}\Phi(t,s)c(s)\,ds + \textit{o}(1),
\end{displaymath}
and the statement i) is proved since $u(t)=1/p(t)$.

\medskip

Secondly, let us assume that $0<\beta^{-}(\mathcal{D})$. By using the statement b) from Proposition \ref{PROP}
we have that $v(t)=\Phi(t,0)v(0)$ and there exists constants $L\geq 1$ and $\eta>0$ such that
$$
|v(0)|\leq Lv(t)e^{-\eta t}=L\Phi(t,0)|v(0)|e^{-\eta t}  \quad \textnormal{for any $t\geq 0$},
$$
or equivalently 
$$
\Phi(t,0)\geq L^{-1}e^{\eta t}  \quad \textnormal{for any $t\geq 0$}.
$$

Now, by using the above estimation combined with the identity (\ref{limits}) 
and the positiveness of $c(t)$ and $\Phi(t,s)$,
it can be deduced that
\begin{displaymath}
\begin{array}{rcl}
p(t) &=& \displaystyle\frac{1}{\displaystyle \Phi(t,0)p^{-1}(0)+\int_{0}^{t}\Phi(t,s)c(s)\,ds} \\\\
&\leq & \displaystyle\frac{1}{\displaystyle L^{-1}p^{-1}(0)e^{\eta t}+\int_{0}^{t}\Phi(t,s)c(s)\,ds}\\\\
&\leq & Lp(0)e^{-\eta t},
\end{array}
\end{displaymath}
and (\ref{limit2}) is verified. As $p(t)$ is non negative, we have that the fraction of occupied patches verifies $p(t)\to 0$ when
$t\to +\infty$ and the statement ii) follows.
\end{proof}

In case that the colonization rate is constant and the extinction rate $t\mapsto e(t)$ is bounded, positive and continuous, the 
Theorem \ref{T1} can be simplified as follows:
\begin{table}[h]
    \centering
    \begin{tabular}{|c|c|}
    \hline
\textnormal{$c< \beta^{-}(e)$}  & \textnormal{$\beta^{+}(e)< c$} \\      
\hline
\textnormal{$p(t)\to 0$ (exponentially)}  & \textnormal{$p(t)=\left(c\int_{0}^{t}\Phi(t,s)\,ds+o(1)\right)^{-1}$} \\
\hline
    \end{tabular}
    \caption{Behavior of the fraction of occupied patches $p(t)$ when $t\to +\infty$ and the colonization rate is constant.}
    \label{tab_1}
\end{table}

Similarly, When the extinction rate is constant and the colonization rate $t\mapsto c(t)$ is bounded, positive and continuous, the 
Theorem \ref{T1} can be simplified as follows:
\begin{table}[h]
    \centering
    \begin{tabular}{|c|c|}
    \hline
\textnormal{$e< \beta^{-}(c)$}   & \textnormal{$\beta^{+}(c)< e$} \\      
\hline
\textnormal{$p(t)=\left(\int_{0}^{t}\Phi(t,s)c(s)\,ds+o(1)\right)^{-1}$}  & \textnormal{$p(t)\to 0$ (exponentially)} \\
\hline
    \end{tabular}
    \caption{Behavior of the fraction of occupied patches $p(t)$ when $t\to +\infty$ and the colonization rate is constant.}
    \label{tab_2}
\end{table}

\section{Other scenarios}
Notice that the case $0\in [\beta^{-}(\mathcal{D}),\beta^{+}(\mathcal{D})]$ is not covered by the previous result
and, according to Remark \ref{consecuencia}, the linear equation (\ref{lin-v}) does not have an exponential dichotomy
on $[0,+\infty)$. In this context, let us recall the following result adapted from \cite[Prop.2.2]{Amster}:
\begin{proposition} 
\label{Prop-Cop}
Let $t\mapsto \mathcal{D}(t)$ be the difference between extinction and colonization
stated in \eqref{D}. The following properties are equivalent:
\begin{itemize}
\item[a)] The linear equation \eqref{lin-v} has an exponential dichotomy on $[0,+\infty)$.
\item[b)] For any 
bounded and continuous function $\mathfrak{v}\colon [0,+\infty) \to \mathbb{R}$, the inhomogenous
equation:
\begin{equation}
\label{inhp}
x'=\mathcal{D}(t)x + \mathfrak{v}(t)    
\end{equation}
has  at least one bounded and continuous solution on $[0,+\infty)$. 
More precisely, $\Sigma(\mathcal{D})\subset (-\infty,0)$ holds if and only if 
all solutions are bounded for each bounded $\mathfrak{v}(\cdot)$   and
$\Sigma(\mathcal{D})\subset (0,+\infty)$ holds if and only if 
there exists exactly one bounded solution for  each bounded $\mathfrak{v}(\cdot)$.
\item[c)] The non homogeneous equation 
\begin{equation}
\label{3nh}
z'=\mathcal{D}(t)z+ 1
\end{equation}
has at least a bounded solution on $[0,+\infty)$. 
\item[d)] There exists a bounded and continuous function $\mathfrak{v}\colon [0,+\infty)\to \mathbb{R}$ with 
$$
\liminf\limits_{t\to +\infty}|\mathfrak{v}(t)|>0
$$
such that the problem \eqref{inhp} has at least one bounded solution. 
\end{itemize}
\end{proposition}

\begin{theorem}
\label{T2}
Let us consider the generalized Levins model described by \eqref{Levins1} where 
$t\mapsto c(t)$ and $t\mapsto e(t)$ are bounded, continuous and positive maps for any $t\in \mathbb{R}$. If $p(0)>0$ is the fraction of the occupied patches at time zero and $0\in [\beta^{-}(d),\beta^{+}(d)]$
then the fraction $t\mapsto p(t)$ of occupied patches at time $t\geq 0$ verifies
\begin{equation}
\label{AB+}
\liminf\limits_{t\to +\infty} p(t)=0.
\end{equation}
\end{theorem}

\begin{proof}
By following the lines to the proof of Theorem \ref{T1}, we know that $t\mapsto p(t)$ with
$0<p(0)<1$ is solution of (\ref{Levins1}) if and only if $t\mapsto u(t)=1/p(t)$ is solution of the inhomogeneous
equation (\ref{Levins2}). Moreover, as we know that $0<p(t)<1$ for any $t\geq 0$, it follows that any solution of (\ref{Levins2}) with $u(0)>1$ is lowerly bounded.

By Remark \ref{consecuencia}, we know that the linear system (\ref{lin-v}) 
does not have an exponential dichotomy on $[0,+\infty)$, then the statement b) of 
Proposition \ref{Prop-Cop} implies that
(\ref{Levins2}) cannot have upperly bounded solutions on $[0,+\infty)$, which leads to
$$
\limsup\limits_{t\to +\infty}u(t)=+\infty \quad \textnormal{which is equivalent to} \quad
\liminf\limits_{t\to +\infty}p(t)=0,
$$
and the Theorem follows.
\end{proof}

\subsection{An example of weakly persistent metapopulation}
In the context of the previous proof, we only have deduced that $t\mapsto u(t)$ is upperly unbounded. Nevertheless,
we point out that the existence of $\ell\geq 1$ such that
\begin{equation}
\label{UP}
\liminf\limits_{t\to +\infty}u(t)=\ell>0
\end{equation}
is theoretically possible and will implies the weak persistence of the metapopulation, namely, $\limsup\limits_{t\to +\infty}p(t)=1/\ell>0$.

In order to provide an example of weak persistence, let us consider a classical metapopulation described by (\ref{Levins1})
having fixed extinction rate $e>0$ and a colonization rate
defined as follows:
\begin{equation}
\label{tam}
c(t)=\left\{\begin{array}{rcl}
e+\Delta & \textnormal{if} &  t\in [T_{2k},T_{2k+1}) \\\\
e-\Delta & \textnormal{if} & t\in [T_{2k+1},T_{2k+2}),
\end{array}\right.
\end{equation}
where $\Delta\in (0,e)$, $p(0)<\Delta/(e+\Delta)$ and the sequence $\{T_{n}\}_{n}$ is described by:
\begin{displaymath}
T_{0}=0 \quad \textnormal{and} \quad T_{1}= \frac{1}{\Delta}
\ln\left(\frac{1}{\varepsilon_{1}}\left[\frac{1}{p(0)}
-\frac{(e+\Delta)}{\Delta}\right]\right)   
\end{displaymath}
and the general times:
\begin{displaymath}
\begin{array}{rcl}
T_{2n}&=& \displaystyle T_{2n-1}+\frac{1}{\Delta} \ln \left(\frac{\frac{1}{\varepsilon_{2n}}
+\frac{e-\Delta}{\Delta}}{\varepsilon_{2n-1}+\frac{2e}{\Delta}}\right)\\\\
T_{2n+1} &=& T_{2n}+\frac{1}{\Delta}
\ln\left(\frac{1}{\varepsilon_{2n+1}}\left[\frac{1}{\varepsilon_{2n}}
-\frac{(e+\Delta)}{\Delta}\right]\right),
\end{array}
\end{displaymath}
where $\varepsilon_{n}$ is a positive sequence convergent to zero satisfying the restriction:
\begin{equation}
\label{ligadura}
\varepsilon_{2n}<\frac{\Delta}{\Delta(1+\varepsilon_{2n-1})+e}.
\end{equation}

As $t\mapsto p^{-1}(t)=u(t)$ is solution of (\ref{Levins2}) 
with $e(t)=e$ and $c(t)$ described by (\ref{tam}), a consequence of
(\ref{ligadura}) is that $t\mapsto u(t)$ is decreasing on $(T_{2n},T_{2n+1})$ whereas is increasing on $(T_{2n-1},T_{2n})$. In addition, it can be verified that
$$
u(T_{2n-1})=\frac{e+\Delta}{\Delta}+\varepsilon_{2n-1} \quad
\textnormal{and} \quad u(T_{2n})=\frac{1}{\varepsilon_{2n}}.
$$

In consequence, if a colonization rate $t\mapsto c(t)$ is described by the function (\ref{tam})
the existence of a solution $t\mapsto p(t)$
of (\ref{Levins1}) satisfying $\limsup\limits_{t\to +\infty}p(t)=\frac{\Delta}{e+\Delta}$
is possible and the weak persistence can be observed.

\end{document}